\documentclass[a4paper,USenglish,numberwithinsect]{lipics-v2019}

\usepackage{mathtools}

\usepackage{tikz}

\usepackage{color}
\definecolor{darkred}{rgb}{0.8,0,0}

\newcommand{\fig}{./fig}
\newcommand{\figref}[1]{\figurename~\ref{#1}}

\newtheorem{observation}[theorem]{Observation}

\newcommand{\cw}{\mathsf{cw}} 

\newcommand{\bn}{\mathsf{b}} 

\newcommand{\dist}{\operatorname{dist}}

\title{Parameterized Complexity of Graph Burning}

\author{Yasuaki Kobayashi}{Kyoto University, Kyoto, Japan}{kobayashi@iip.ist.i.kyoto-u.ac.jp}{0000-0003-3244-6915}
{JSPS KAKENHI Grant Number JP20K19742.}

\author{Yota Otachi}{Nagoya University, Nagoya, Japan}{otachi@nagoya-u.jp}{0000-0002-0087-853X}
{JSPS KAKENHI Grant Numbers JP18K11168, JP18K11169, JP18H04091.}

\authorrunning{Y. Kobayashi, Y. Otachi}

\Copyright{Yasuaki Kobayashi, Yota Otachi}



\category{}

\relatedversion{}

\supplement{}

\funding{}

\acknowledgements{The authors would like to thank Michael Lampis for bringing the \textsc{Graph Burning} 
problem to their attention.}

\nolinenumbers 

\hideLIPIcs  

\EventEditors{Yixin Cao and Marcin Pilipczuk}
\EventNoEds{2}
\EventLongTitle{15th International Symposium on Parameterized and Exact Computation (IPEC 2020)}
\EventShortTitle{IPEC 2020}
\EventAcronym{IPEC}
\EventYear{2020}
\EventDate{December 14--18, 2020}
\EventLocation{Hong Kong, China (Virtual Conference)}
\EventLogo{}
\SeriesVolume{180}
\ArticleNo{21}

\begin{document}
\maketitle

\begin{abstract}
\textsc{Graph Burning} asks, given a graph $G = (V,E)$ and an integer $k$,
whether there exists $(b_{0},\dots,b_{k-1}) \in V^{k}$ such that every vertex in $G$ has distance at most $i$ from some $b_{i}$.
This problem is known to be NP-complete even on connected caterpillars of maximum degree $3$.
We study the parameterized complexity of this problem 
and answer all questions arose by Kare and Reddy~[IWOCA 2019] about parameterized complexity of the problem.
We show that the problem is W[2]-complete parameterized by $k$
and that it does not admit a polynomial kernel parameterized by vertex cover number unless $\mathrm{NP} \subseteq \mathrm{coNP/poly}$.
We also show that the problem is fixed-parameter tractable parameterized 
by clique-width plus the maximum diameter among all connected components.
This implies the fixed-parameter tractability parameterized by modular-width, by treedepth, and by distance to cographs.
Although the parameterization by distance to split graphs cannot be handled with the clique-width argument,
we show that this is also tractable by a reduction to a generalized problem with a smaller solution size.
\end{abstract}

\section{Introduction}
Bonato, Janssen, and Roshanbin~\cite{BonatoJR14,BonatoJR16} introduced \textsc{Graph Burning} as a model of information spreading.
This problem asks to burn all the vertices in a graph in the following way:
we first pick a vertex and set fire to the vertex;
at the beginning of each round, the fire spreads one step along edges;
at the end of each round, we pick a vertex and set fire to it;
the process finishes when all vertices are burned.
The objective in the problem is to minimize the number of rounds
(including the first one for just picking the first vertex) to burn all the vertices.
The minimum number of rounds that can burn a graph $G$ in such a process is 
the \emph{burning number} of $G$, which is denoted by $\bn(G)$.
Given a graph $G$ and an integer $k$, \textsc{Graph Burning} asks whether $\bn(G) \le k$.

In other words, the burning number of $G$ can be defined as the minimum length $k$ of 
a sequence $(b_{0}, \dots, b_{k-1})$ of vertices of $G$ such that 
every vertex in $G$ has distance at most $i$ from some $b_{i}$.
We call such a sequence a \emph{burning sequence}.
Note that in this definition, $b_{i}$ is the vertex we set fire in the $(k-i)$th round.
Note also that we do not ask the vertices $b_{0}, \dots, b_{k-1}$ to be distinct.
It is also useful to introduce the generalized neighborhood of vertices.
For a vertex $v$ of a graph $G$, 
let $N_{d}[v]$ be the set of vertices with distance at most $d$ in $G$.
For example, $N_{0}[v]$ contains only $v$,
$N_{1}[v]$ is just the closed neighborhood of $v$,
and $N_{2}[v] = \bigcup_{u \in N[v]} N[u]$.
With this terminology,
a sequence $(b_{0}, \dots, b_{k-1})$ of vertices of $G = (V,E)$ is a burning sequence of $G = (V,E)$
if and only if $\bigcup_{0 \le i \le k-1}N_{i}[b_{i}] = V$.

\subsection{Previous work}

As a model of information spreading, 
it is important to know how fast the information can spread under the model in the worst case.
This question can be answered by finding the maximum burning number of graphs of $n$ vertices.
The literature is rich in this direction.
In the very first paper~\cite{BonatoJR14,BonatoJR16},
it is shown that $\bn(G) \le 2 \lceil \sqrt{n} \rceil -1$ for every connected graph $G$ of order $n$
and conjectured that $\bn(G) \le \lceil \sqrt{n} \rceil$ holds.
Note that the connectivity requirement is essential here as an edgeless graph of order $n$ needs $n$ rounds.
Some improvements of the general upper bound 
and studies on special cases are done~\cite{LandL16,MitschePR17,FitzpatrickW17arxiv,BessyBJRR18,DasDSSS18,MitschePR18,BonatoEKM19arxiv,BonatoL19,LiuZH19,Hiller19arxiv,BonatoGS20,LiuHH20a,LiuHH20b,TanT20,FarokhTTB20arxiv},
but the conjecture of $\bn(G) \le \lceil \sqrt{n} \rceil$ for general connected graphs remains unsettled.
The current best upper bound is $\lceil (-3 + \sqrt{24n+ 33})/4 \rceil \approx \sqrt{1.5 n}$~\cite{LandL16}.

The computational complexity of \textsc{Graph Burning} has been studied intensively as well.
It is shown that \textsc{Graph Burning} is NP-complete on trees of maximum degree 3, spiders, and linear forests~\cite{BessyBJRR17}.
The NP-completeness result is further extended to connected caterpillars of maximum degree $3$~\cite{LiuHH20a},
which form subclasses of connected interval graphs, connected permutation graphs, and connected unit disk graphs.
On the other hand, \textsc{Graph Burning} admits a 3-approximation algorithm for general graphs~\cite{BonatoK19};
that is, given a graph $G$, the algorithm finds a burning sequence of $G$ of length at most $3 \cdot \bn(G)$ in polynomial time.
Algorithms with approximation factors parameterized by path-length and tree-length are known as well~\cite{KamaliMZ20}.
Bonato and Kamali~\cite{BonatoK19} asked whether the problem is APX-hard.
Recently, this question has been answered in the affirmative~\cite{MondalPKR20arxiv}.

Kare and Reddy~\cite{KareR19} initiated the study on parameterized complexity of \textsc{Graph Burning}.
They showed that \textsc{Graph Burning} on connected graphs is fixed-parameter tractable parameterized
by distance to cluster graphs (disjoint unions of complete graphs) and by neighborhood diversity.
The parameterized complexity with respect to the natural parameter $k$, the burning number, remained open.
Recently, Janssen~\cite{Janssen20arxiv} has generalized the problem to directed graphs and has shown that 
the directed version is W[2]-complete parameterized by $k$ even on directed acyclic graphs.
It was mentioned in~\cite{Janssen20arxiv} that the original undirected version parameterized by $k$ was still open.

For further information about the previous studies, 
see the recent comprehensive survey by Bonato~\cite{Bonato20arxiv}.

\subsection{Our results}
In the literature, the input graph of \textsc{Graph Burning} is sometimes assumed to be connected (e.g.\ in~\cite{KareR19}).
In this paper, however, we do not generally assume the connectivity of input graphs.
All positive results in this paper hold on possibly disconnected graphs,
while all negative results hold even on connected graphs.
Note that in \textsc{Graph Burning}, the disconnected case would be nontrivially more complex than the connected case.
For example, while the burning number of a path of $n$ vertices is $\lceil \sqrt{n} \rceil$~\cite{BonatoJR14,BonatoJR16},
\textsc{Graph Burning} is NP-complete on graphs obtained as the disjoint union of paths~\cite{BessyBJRR17}.

Our study in this paper is inspired by Kare and Reddy~\cite{KareR19} and Janssen~\cite{Janssen20arxiv}.
We generalize the results in~\cite{KareR19} and solve all open problems on parameterized complexity in~\cite{KareR19}.
In Section~\ref{sec:tractable}, we present our positive algorithmic results.
We show that \textsc{Graph Burning} is fixed-parameter tractable parameterized
by clique-width plus the maximum diameter among all connected components.
This implies that \textsc{Graph Burning} is fixed-parameter tractable parameterized
by modular-width, by treedepth, and by distance to cographs.
We also show that \textsc{Graph Burning} is fixed-parameter tractable parameterized by distance to split graphs.
The complexity parameterized by distance to cographs and by distance to split graphs are explicitly asked in~\cite{KareR19}.
The fixed-parameter tractability parameterized by modular-width generalizes 
the one parameterized by neighborhood diversity in~\cite{KareR19}.
In Section~\ref{sec:intractable}, we present some negative results.
We show that \textsc{Graph Burning} parameterized by the natural parameter $k$ is W[2]-complete.
This settles the main open problem in this line of research~\cite{KareR19,Janssen20arxiv}.
As a byproduct, we also show that \textsc{Graph Burning} parameterized by vertex cover number 
does not admit a polynomial kernel unless $\mathrm{NP} \subseteq \mathrm{coNP/poly}$.
This also answers a question in~\cite{KareR19}.
See \figref{fig:width-parameters} for a summary of the results.

\begin{figure}[htb]
  \centering
  \begin{tikzpicture}[]
    \node (widthparas) at (0,0) {\includegraphics[width=\textwidth]{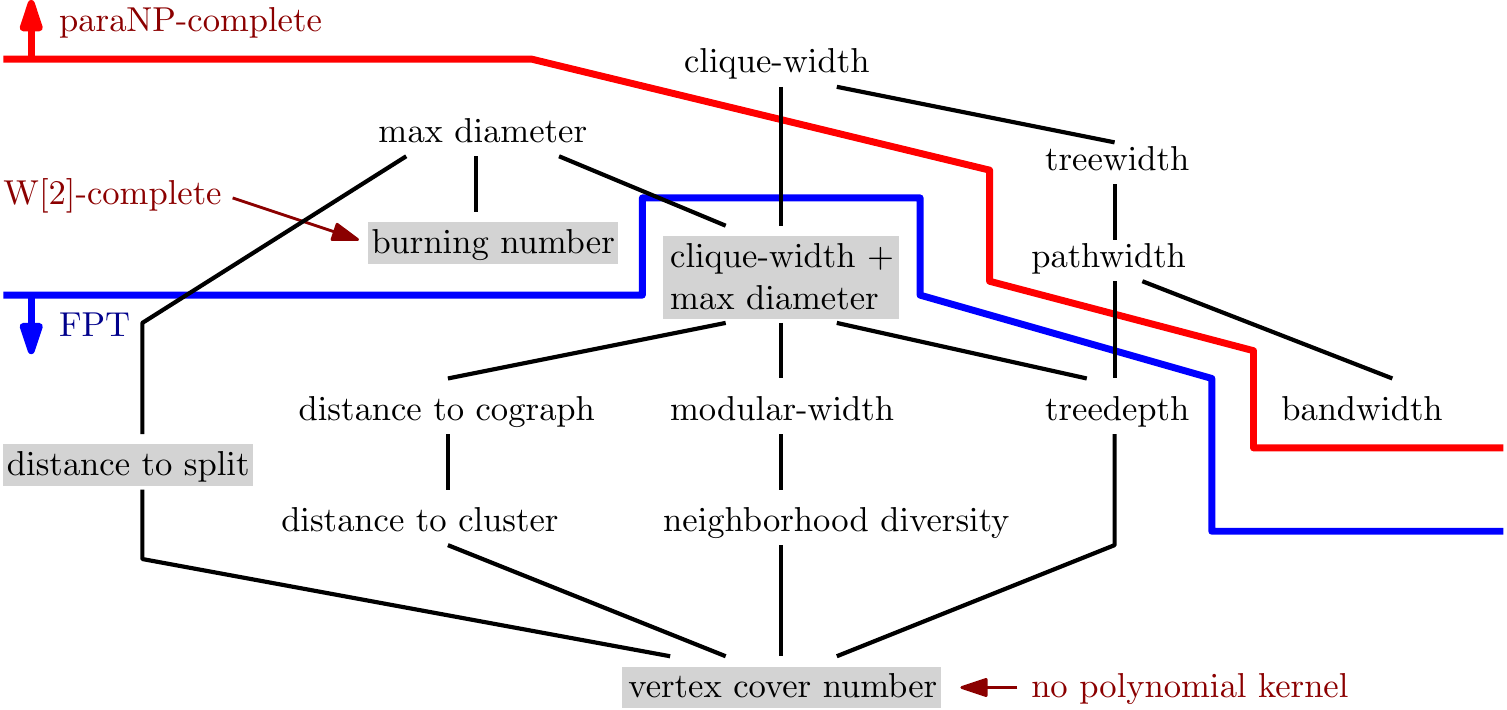}};
    \node (nd)  at (2.7,-1.55) {\footnotesize \cite{KareR19}};
    \node (dtc) at (-1.45,-1.55) {\footnotesize \cite{KareR19}};
    \node (bw)  at (6.75,-.53) {\footnotesize \cite{LiuHH20a}};
  \end{tikzpicture}
  \caption{Graph parameters and the complexity of \textsc{Graph Burning}.
    The results on the parameters with dark background are the main results of the paper.
    Connections between two parameters imply the existence of a function in the one above 
    (being in this sense more general) that lowerbounds the one below.    
    ``The maximum diameter among all connected components'' is shortened as ``max diameter.''
    The paraNP-completeness parameterized by bandwidth follows from the result 
    by Liu et al.~\cite{LiuHH20a} who showed that the problem is NP-complete 
    on caterpillars of maximum degree $3$, which have bandwidth at most $2$.}
  \label{fig:width-parameters}
\end{figure}

We assume that the readers are familiar with the basic terms and concepts in
the parameterized complexity theory. 
See some textbooks in the field (e.g., \cite{DowneyF99,CyganFKLMPPS15}) for definitions.
We omit the definitions of most of the graph parameters in this paper as we do not explicitly need them.
We only need the definition of the \emph{vertex cover number} of a graph:
it is the minimum integer $k$ such that their exists a set of $k$ vertices of the graph (called a \emph{vertex cover})
such that each edge in the graph has at least one endpoint in the set.
We refer the readers to~\cite{SorgeW19} for the definitions of other graph parameters and the hierarchy among them.


\section{Positive results}
\label{sec:tractable}

We first observe that \textsc{Graph Burning} is 
expressible as a first order logic (FO) formula of length depending only on $k$.

The syntax of FO of graphs includes
(i) the logical connectives $\lor$, $\land$, $\lnot$, $\Leftrightarrow$, $\Rightarrow$,
(ii) variables for vertices,
(iii) the quantifiers $\forall$ and $\exists$ applicable to these variables, and
(iv) the following binary relations:
equality of variables, and
$\mathsf{adj}(u,v)$ for two vertex variables $u$ and $v$,
which means that $u$ and $v$ are adjacent.
If $G$ models an FO formula $\varphi$ with no free variables,
then we write $G \models \varphi$.

The following formula $\dist_{\le d}(v,w)$ is true
if and only if the distance between $v$ and $w$ is at most $d$:
\[
  \dist_{\le d}(v,w) \coloneqq \exists u_{0}, \dots, u_{d} \;
  (u_{0} = v) \land (u_{d} = w) \land
  \left(\bigwedge_{0 \le i < d} (u_{i} = u_{i+1}) \lor \mathsf{adj}(u_{i},u_{i+1})\right).
\]
Clearly, $\dist_{\le d}(v,w)$ has length depending only on $d$.
Now we define the formula $\varphi_{k}$
such that $G \models \varphi_{k}$ if and only if $(G,k)$ is a yes instance of \textsc{Graph Burning}
as follows:
\[
  \varphi_{k} \coloneqq \exists v_{0}, \dots, v_{k-1} \; \forall v \; 
  \bigvee_{0 \le i \le k-1} \dist_{\le i}(v, v_{i}).
\]

It is known that on nowhere dense graph classes,
testing an FO formula $\psi$ is fixed-parameter tractable parameterized by $|\psi|$,
where $|\psi|$ is the length of $\psi$~\cite{GroheKS17}.
(See \cite{GroheKS17} for the definition of nowhere dense graph classes.) 
Since $\varphi_{k}$ is an FO formula of length depending only on $k$,
the following holds.
\begin{observation}
  \textsc{Graph Burning} on nowhere dense graphs
  parameterized by $k$ is fixed-parameter tractable.
\end{observation}

For an $n$-vertex graph $G$ of clique-width at most $\cw$
and for a one-sorted monadic-second order logic (MSO$_{1}$) formula $\psi$,
one can check whether $G \models \psi$
in time $O(f(|\psi|, \cw) \cdot n^{3})$,
where $f$ is a computable function~\cite{CourcelleMR00,Oum08}.
Since an FO formula is an MSO$_{1}$ formula
and the length of $\varphi_{k}$ depends only on $k$,
we can observe the following fact.
\begin{observation}
  \label{obs:cw+k}
  \textsc{Graph Burning} parameterized by clique-width${}+k$ is fixed-parameter tractable.
\end{observation}
We extend this observation in a nontrivial way to show the main result of this section.
To this end, it is useful to generalize $\varphi_{k}$ as follows.
For nonempty $I = \{i_{1}, \dots, i_{|I|}\} \subseteq \{0, \dots, k-1\}$,
let $\varphi_{I}$ be a formula that means that 
there are vertices $b_{i_{1}}, \dots, b_{i_{|I|}}$ such that
$\bigcup_{i_{j} \in I}N_{i_{j}}[b_{i_{j}}] = V$, which can be expressed as follows:
\[
  \varphi_{I} \coloneqq 
  \exists v_{i_{1}}, \dots, v_{i_{|I|}} \; \forall v \; 
  \bigvee_{1 \le j \le |I|} \dist_{\le i_{j}}(v, v_{i_{j}}).
\]
Clearly, checking $G \models \varphi_{I}$ is still fixed-parameter tractable
parameterized by clique-width${}+k$.

\begin{theorem}
\label{thm:fpt-cw+maxdiam}
\textsc{Graph Burning} is fixed-parameter tractable parameterized
by the clique-width of the input graph plus the maximum diameter among all connected components.
\end{theorem}
\begin{proof}
Let $(G,k)$ be an instance of \textsc{Graph Burning}, where $G$ has $n$ vertices.
Let $C_{1}, \dots, C_{p}$ be the connected components of $G$
and $d_{\max}$ be the maximum diameter of the components.
We assume that $k \ge p$ since otherwise $(G,k)$ is a trivial no instance.
By Observation~\ref{obs:cw+k}, we can also assume that $d_{\max} < k$.

Observe that if a component $C_{q}$ contains $b_{i}$ for some $i \ge d_{\max}$, then $N_{i}[b_{i}] = V(C_{q})$.
Hence the problem is equivalent to finding a sequence $(b_{0}, \dots, b_{d_{\max}-1})$
that burns as many connected components as possible.
That is, we want to find a maximum cardinality subset $\mathcal{C} \subseteq \{C_{1}, \dots, C_{p}\}$
and a sequence $(b_{0}, \dots, b_{d_{\max}-1})$
such that $\bigcup_{0 \le i \le d_{\max}-1} N_{i}[b_{i}] = \bigcup_{C_{q} \in \mathcal{C}} V(C_{q})$.
It holds that $p-|\mathcal{C}| \le k - d_{\max}$ if and only if $(G,k)$ is a yes instance.

We reduce this problem to \textsc{Disjoint Sets}.
Given a universe $U$, a subset family $\mathcal{S} \subseteq 2^{U}$, and an integer $t$,
\textsc{Disjoint Sets} asks whether there are $t$ pairwise-disjoint subsets in $\mathcal{S}$.
Let $U=\{0,1,\dots,d_{\max}-1\} \cup \{c_{1}, \dots, c_{p}\}$ and 
$\mathcal{S} = \{I \cup \{c_{q}\} \mid I \subseteq \{0,\dots,d_{\max}-1\}, \; C_{q} \models \varphi_{I}, \; 1 \le q \le p\}$.
Clearly, picking $I \cup \{c_{q}\}$ into the solution for the \textsc{Disjoint Set} instance corresponds to 
burning $C_{q}$ with $\{b_{i} \mid i \in I\}$, and vice versa.
We set $t = p- k +d_{\max}$.
Note that $t \le d_{\max}$ as $k \ge p$.
Using the color-coding technique,
it can be shown that \textsc{Disjoint Sets} is solvable in time 
$2^{O(t \cdot \max_{S \in \mathcal{S}} |S|)} (|\mathcal{S}|+|U|)^{O(1)}$~\cite[\textsc{Disjoint $r$-Subsets}]{DowneyF99}
(see also~\cite[\textsc{Bounded Rank Disjoint Sets}]{DomLS14}).
In our instance, $t \cdot \max_{S \in \mathcal{S}} |S| \le d_{\max}(d_{\max}+1)$ holds.

For each $q \in \{1,\dots,p\}$ and for each $I \subseteq \{0,\dots,d_{\max}-1\}$,
we can test whether $C_{q} \models \varphi_{I}$
in time $O(f(d_{\max}+\cw) \cdot n^{3})$ for some computable function $f$,
where $\cw$ is the clique-width of $G$,
because $|\varphi_{I}|$ depends only on $d_{\max}$.
Thus, $\mathcal{S}$ can be constructed in time $O(f(d_{\max}+\cw) \cdot n^{3} \cdot p \cdot 2^{d_{\max}})$.
Since $|\mathcal{S}| \le 2^{d_{\max}} \cdot p$,
the last step of solving \textsc{Disjoint Sets} can be done in time
$2^{O(d_{\max}^{2})} (2^{d_{\max}} \cdot p + (d_{\max}+p))^{O(1)}$.
Since $p \le n$, the total running time is $g(d_{\max} + \cw) \cdot n^{O(1)}$ for some computable function $g$.
This completes the proof.
\end{proof}

The definition of modular-width implies that every connected component of a graph of modular-width at most $w$
has both diameter and clique-width at most $w$~\cite{GajarskyLO13}.
It is known that every connected component of a graph of treedepth at most $d$ has diameter at most $2^{d}$~\cite{NesetrilO2012}
and its clique-width is bounded by a function of treedepth (or even smaller treewidth)~\cite{CorneilR05}.
Therefore, Theorem~\ref{thm:fpt-cw+maxdiam} implies the fixed-parameter tractability with respect to these parameters.
\begin{corollary}
\label{cor:fpt-mw}
\textsc{Graph Burning} is fixed-parameter tractable parameterized by modular-width.
\end{corollary}
\begin{corollary}
\label{cor:fpt-td}
\textsc{Graph Burning} is fixed-parameter tractable parameterized by treedepth.
\end{corollary}

For a graph class $\mathcal{C}$ and a graph $G$, 
the \emph{distance} from $G$ to $\mathcal{C}$ is defined as the minimum integer $k$
such that by removing at most $k$ vertices from $G$, one can obtain a member of $\mathcal{C}$.

Let $\mathcal{C}$ be a graph class with constants $c$ and $d$ such that
each graph in $\mathcal{C}$ has clique-width at most $c$ and
each connected component of each member of $\mathcal{C}$ has diameter at most $d$.
Observe that a graph of distance at most $k$ to $\mathcal{C}$ has clique-width at most $c \cdot 2^{k}$
since after a removal of a single vertex, the clique-width remains
at least half of the original clique-width~\cite{Gurski17}.
Observe also that each connected component of a graph of distance at most $k$ to $\mathcal{C}$
has diameter at most $k + d$.
Thus, by Theorem~\ref{thm:fpt-cw+maxdiam},
\textsc{Graph Burning} is fixed-parameter tractable parameterized by distance to $\mathcal{C}$.
This observation can be applied immediately to cographs that are known to be the $P_{4}$-free graphs
and the graphs of clique-width at most $2$.
\begin{corollary}
\textsc{Graph Burning} is fixed-parameter tractable parameterized by distance to cographs.
\end{corollary}

Now we consider the distance to split graphs.
A graph is a \emph{split graph} if its vertex set can be partitioned into a clique and an independent set.
For a graph $G = (V,E)$, which is not necessarily a split graph,
a subset $S \subseteq V$ is a \emph{split-deletion set} if $G - S$ is a split graph.
Then the distance to split graphs from $G$ is equal to the minimum size of a split-deletion set.
It is known that the split graphs are exactly the $(2K_{2}, C_{4}, C_{5})$-free graphs~\cite{FoldesH77}.
This characterization implies that when designing an algorithm parameterized by distance $d$ to split graphs,
we can assume that a split-deletion set of minimum size is given
since a standard bounded-search tree algorithm finds such a set in time $5^{d} \cdot n^{O(1)}$,
where the number $5$ is the maximum order of the forbidden induced subgraphs.

\begin{theorem}
\textsc{Graph Burning} is fixed-parameter tractable parameterized by distance to split graphs.
\end{theorem}
\begin{proof}
Let $(G,k)$ be an instance of \textsc{Graph Burning} and
$S$ be a minimum split-deletion set of $G = (V,E)$. We denote $|S|$ by $s$.
Let $(K,I)$ be a partition of $V-S$, where $K$ is a clique and $I$ is an independent set.
Such a partition can be found in linear time by greedily adding a vertex of minimum degree into $I$.
We further partition $I$ into $I_{K}$, $I_{S}$, and $I_{\emptyset}$ in such a way that
$I_{\emptyset}$ is the set of degree-$0$ vertices in $G$,
$I_{S}$ is the set of vertices in $I \setminus I_{\emptyset}$ that have neighbors only in $S$,
and $I_{K} = I \setminus (I_{\emptyset} \cup I_{S})$.
Note that each vertex in $I_{S}$ has at least one neighbor in $S$
and each vertex in $I_{K}$ has at least one neighbor in $K$
(and possibly some neighbors in $S$).

We first reduce the number of vertices in $I_{S}$.
For a nonempty subset $S' \subseteq S$, let $J_{S'} \subseteq I_{S}$
be the set of vertices whose neighborhood is exactly $S'$.
Since the pairwise distance between vertices in $J_{S'}$ is $2$,
a burning sequence does not pick four or more vertices in $J_{S'}$.
Thus, we can remove all but three vertices in $J_{S'}$ and obtain an equivalent instance.
We apply this reduction to all subsets $S' \subseteq S$
and denote the reduced subset of $I_{S}$ by $I^{*}_{S}$.
Note that $|I^{*}_{S}| < 3 \cdot 2^{|S|}$.
If $k \ge 3 + |S| + |I^{*}_{S}| + |I_{\emptyset}| = 3 + s + 3 \cdot 2^{s} + |I_{\emptyset}|$,
then $(G,k)$ is a yes instance:
we can take all the vertices in $S \cup I^{*}_{S} \cup I_{\emptyset}$ and three vertices in $K \cup I_{K}$
and then arbitrarily order them to construct a burning sequence of $G$.
Hence, in the following, we assume that $k < 3 + s + 3 \cdot 2^{s} + |I_{\emptyset}|$.

We next remove all vertices in $I_{\emptyset}$ and obtain an equivalence instance of a slightly generalized problem.
Observe that if $(G,k)$ is a yes instance,
then $k \ge |I_{\emptyset}|$ and there is a burning sequence $(b_{0},\dots,b_{k-1})$ of $G$
such that the set of first $|I_{\emptyset}|$ vertices $\{b_{0}, \dots, b_{|I_{\emptyset}|}\}$ is $I_{\emptyset}$:
we need to take every isolated vertex into a burning sequence,
but even $b_{0}$ is good enough to burn an isolated component.
In the following, we only consider burning sequences with this restriction.
Now the problem is reduced to the one for finding a sequence 
$(b_{|I_{\emptyset}|}, b_{|I_{\emptyset}|+1},\dots,b_{k-1})$
of vertices in $V \setminus I_{\emptyset}$
such that $\bigcup_{|I_{\emptyset}| \le i \le k-1} N_{i}[b_{i}] = V  \setminus I_{\emptyset}$.
We denote by $(G',k,|I_{\emptyset}|)$ the obtained instance of the new problem, 
where $G' = G[K \cup I_{K} \cup S \cup I^{*}_{S}]$.

To solve the reduced problem, we first guess which vertices in $S \cup I^{*}_{S}$
appear in $(b_{|I_{\emptyset}|}, b_{|I_{\emptyset}|+1},\dots,b_{k-1})$
and where they are placed in the sequence.
The number of candidates of such a guess depends only on $s$
as $|S \cup I^{*}_{S}|^{k-|I_{\emptyset}|} < (s+3\cdot 2^{s})^{s+3\cdot 2^{s}+3}$.
The vacant slots of $(b_{|I_{\emptyset}|}, b_{|I_{\emptyset}|+1},\dots,b_{k-1})$
after the guess tell us which $b_{i}$ belongs to $K \cup I_{K}$.
If there are at most three vertices in $K \cup I_{K}$ appear 
in $(b_{|I_{\emptyset}|}, b_{|I_{\emptyset}|+1},\dots,b_{k-1})$
then we try all $O(n^{3})$ combinations to complete the sequence.
Otherwise, we guess from $O(n)$ candidates the vertex in $K \cup I_{K}$
that appears in $(b_{|I_{\emptyset}|}, b_{|I_{\emptyset}|+1},\dots,b_{k-1})$ and has the largest index.
Since the index of the guessed vertex in $(b_{|I_{\emptyset}|}, b_{|I_{\emptyset}|+1},\dots,b_{k-1})$ is at least $3$
and $K \cup I_{K}$ induces a connected split graph, which has diameter at most $3$,
the guessed vertex in $K \cup I_{K}$ burns all vertices in $K \cup I_{K}$.

Finally we fill the positions in $(b_{|I_{\emptyset}|}, b_{|I_{\emptyset}|+1},\dots,b_{k-1})$ that still remain vacant.
Let $X \subseteq \{|I_{\emptyset}|, \dots, k-1\}$ be the set of indices $i$ for which no vertex is guessed as $b_{i}$ so far,
and let $\overline{X} = \{|I_{\emptyset}|, \dots, k-1\} \setminus X$.
Let $U = V(G') \setminus (\bigcup_{i \in \overline{X}} N_{i}[b_{i}])$.
Note that $U \subseteq S \cup I^{*}_{S}$.
Our task is to find $\{b_{i} \mid i \in X\} \subseteq K \cup I_{K}$ such that
$U \subseteq \bigcup_{i \in X} N_{i}[b_{i}]$.
This task can be seen as an instance $(U', \mathcal{S}, p)$ of \textsc{Set Cover}, where 
$U' = U \cup X$, 
$\mathcal{S} = \{(N_{i}[v] \cap U) \cup \{i\} \mid v \in K \cup I_{K}, i \in X\}$, and
$p = |X| < k - |I_{\emptyset}| \le 3 + s + 3 \cdot 2^{s}$.
Since \textsc{Set Cover} parameterized by $|U'|$ is fixed-parameter tractable~\cite[Lemma~2]{FominKW04}
and $|U'| \le s + 3 \cdot 2^{s} + p$, the theorem follows.
\end{proof}


\section{Negative results}
\label{sec:intractable}

This section is devoted to the proofs of the following theorems.

\begin{theorem}
\label{thm:w2-complete-k}
\textsc{Graph Burning} is $\mathrm{W[2]}$-complete parameterized by $k$.
\end{theorem}

\begin{theorem}
\label{thm:no-poly-kernel-vc}
\textsc{Graph Burning} does not admit a polynomial kernel parameterized by vertex cover number
unless $\mathrm{NP} \subseteq \mathrm{coNP/poly}$.
\end{theorem}

We present a reduction from \textsc{Set Cover} to \textsc{Graph Burning}
that proves both Theorems~\ref{thm:w2-complete-k} and \ref{thm:no-poly-kernel-vc}.
Given a set $U = \{u_{1}, \dots, u_{n}\}$,
a family of nonempty subsets $\mathcal{S} = \{S_{1}, \dots, S_{m}\} \subseteq 2^{U} \setminus \{\emptyset\}$,
and a positive integer $s$,
\textsc{Set Cover} asks whether there exists a subfamily $\mathcal{S}' \subseteq \mathcal{S}$
such that $|\mathcal{S}'| \le s$ and $\bigcup_{S \in \mathcal{S}'} S = U$.

Let $(U = \{u_{1}, \dots, u_{n}\}, \mathcal{S} = \{S_{1}, \dots, S_{m}\}, s)$ be an instance of \textsc{Set Cover}. 
We construct an equivalent instance $(G, k = s+2)$ of \textsc{Graph Burning}.
(See \figref{fig:reduction}.)
We first construct $s = k-2$ isomorphic graphs $G_{2}, \dots, G_{k-1}$ as follows.
For each $i \in \{2,3,\dots,k-1\}$,
the vertex set of $G_{i}$ is $U_{i} \cup V_{i}$,
where $U_{i} = \{u_{1}^{(i)}, \dots, u_{n}^{(i)}\}$ is a clique
and $V_{i} = \{v_{1}^{(i)}, \dots, v_{m}^{(i)}\}$ is an independent set.
In $G_{i}$, $u_{p}^{(i)}$ and $v_{q}^{(i)}$ are adjacent if and only if $u_{p} \in S_{q}$.
From each $G_{i}$, we construct $H_{i}$
by adding $i+2$ copies of a path of $i$ vertices
and all possible edges between each vertex in $V_{i}$
and one of the degree-1 vertices in each path.
We then take the disjoint union of $H_{2}, H_{3}, \dots, H_{k-1}$
and add $U$ as a clique.
For each $i \in \{2,3,\dots,k-1\}$ and $j \in \{1,2,\dots,m\}$,
we connect $u_{j}^{(i)}$ and $u_{j}$ with a path of length $i-1$
with $i-2$ new inner vertices.
Finally, we attach a vertex $w$ to a vertex in $U$,
and a path $(x,y,z)$ to the same vertex.
We set $V_{0} = \{w\}$ and $V_{1} = \{x,y,z\}$.
We denote the constructed graph by $G$.

\begin{figure}[htb]
  \centering
  \includegraphics[width=\textwidth]{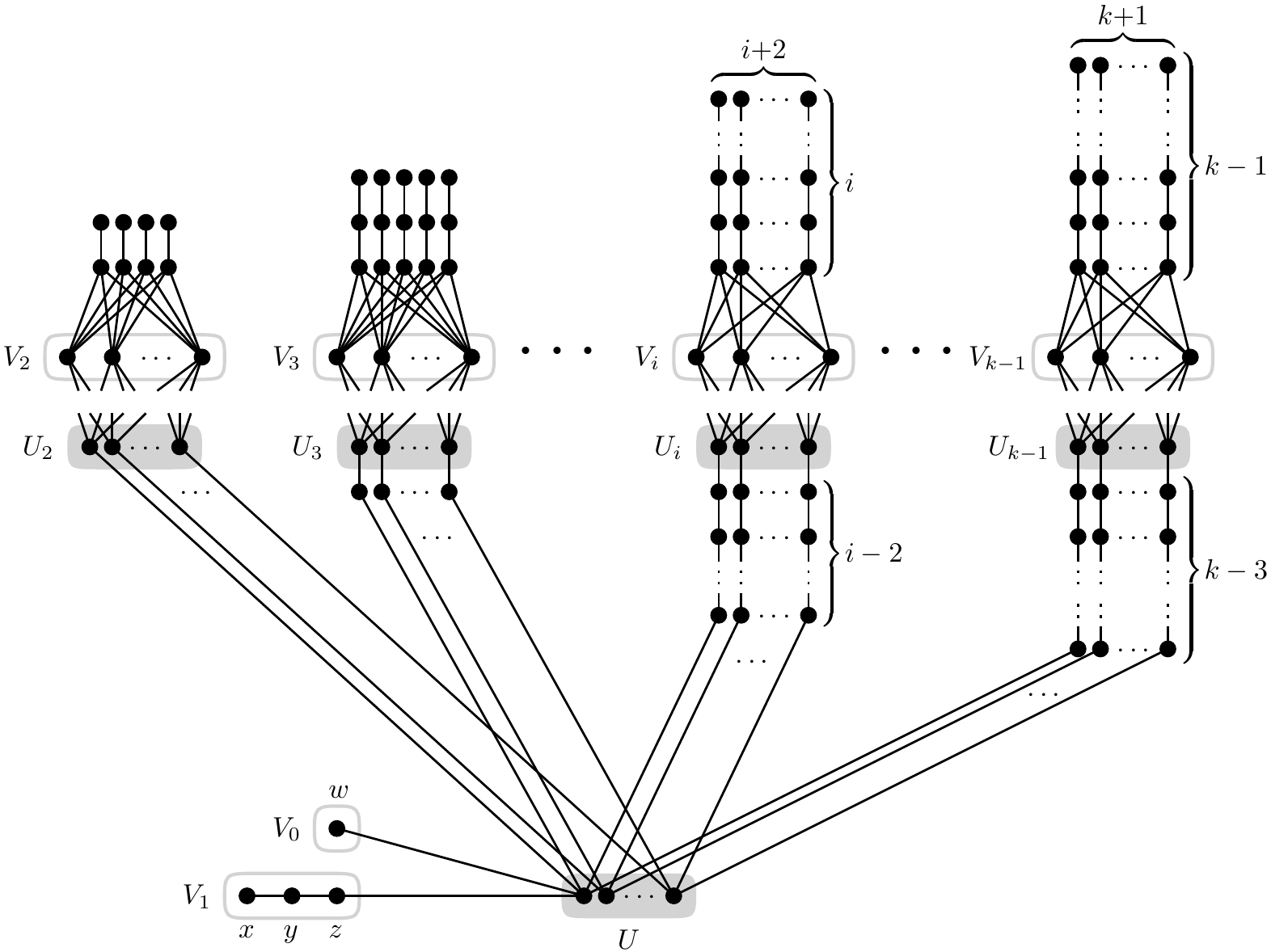}
  \caption{The reduction from \textsc{Set Cover} to \textsc{Graph Burning}.
    The edges in the cliques $U_{2}, \dots, U_{k-1}$, and $U$ are omitted.}
  \label{fig:reduction}
\end{figure}

\begin{lemma}
\label{lem:sc-to-gb}
$(U, \mathcal{S}, s)$ is a yes instance of \textsc{Set Cover}
if and only if
$(G,k)$ is a yes instance of \textsc{Graph Burning}.
\end{lemma}
\begin{proof}
($\implies$)
Assume that $(U, \mathcal{S}, s)$ is a yes instance of \textsc{Set Cover}
and $\mathcal{S}' \subseteq \mathcal{S}$ is a certificate;
that is, $|\mathcal{S}'| \le s$ and $\bigcup_{S \in \mathcal{S}'} S = U$.
We assume without loss of generality that $|\mathcal{S}'| = s$
and $\mathcal{S}' = \{S_{2}, S_{3}, \dots, S_{s+1 = k-1}\}$.
We set $b_{0} = w$, $b_{1} = y$, and $b_{i} = v_{i}^{(i)}$ for $2 \le i \le k-$1.
We show that $(b_{0}, \dots, b_{k-1})$ is a burning sequence of $G$.

Clearly, $N_{0}[b_{0}] = V_{0}$ and $N_{1}[b_{1}] = V_{1}$.
For $2 \le i \le k-1$, observe that $N_{i}[b_{i}] = N_{i}[v_{i}^{(i)}]$
includes all the vertices of $H_{i}$:
the farthest vertices in $i$-vertex paths have distance exactly $i$ from $v_{i}^{(i)}$;
$\dist(v_{i}^{(i)}, v_{j}^{(i)}) = 2$ for each $j \ne i$
as $v_{j}^{(i)}$ and $v_{i}^{(i)}$ share a neighbor (an endpoint of a path of $i$ vertices);
$\dist(v_{i}^{(i)}, u_{j}^{(i)}) \le 2$ for every $j$ 
since $v_{i}^{(i)}$ has at least one neighbor in the clique $U_{i}$ as $\emptyset \notin \mathcal{S}$.
Moreover, $N_{i}[v_{i}^{(i)}]$ includes all inner vertices of the paths from $U_{i}$ to $U$
as $\dist(v_{i}^{(i)}, u_{j}^{(i)}) \le 2$ for every $j$.
Finally, $u_{j} \in N_{i}[v_{i}^{(i)}]$ if and only if $u_{j} \in S_{i}$:
if $u_{j} \in S_{i}$, then $v_{i}^{(i)}$ and $u_{j}^{(i)}$ are adjacent, and thus
$\dist(v_{i}^{(i)}, u_{j}) \le 1 + \dist(u_{j}^{(i)}, u_{j}) = i$;
otherwise, $v_{i}^{(i)}$ and $u_{j}^{(i)}$ are not adjacent and thus
$\dist(v_{i}^{(i)}, u_{j}) \ge \min\{2 + \dist(u_{j}^{(i)}, u_{j}), 1 + \dist(u_{h \ne j}^{(i)}, u_{j})\} > i$.
This implies that $U \subseteq \bigcup_{2 \le i \le k-1} N_{i}[v_{i}^{(i)}]$
since $\bigcup_{S \in \mathcal{S}'} S = U$.

\medskip

($\impliedby$)
Assume that $(G,k)$ is a yes instance of \textsc{Graph Burning}
and $(b_{0}, \dots, b_{k-1})$ is a burning sequence of $G$.

We first show that $b_{i} \in V_{i}$ for all $0 \le i \le k-1$. Let $i \in \{2,\dots,k-1\}$. 
Assume that we already know that $b_{j} \in V_{j}$ for $i+1 \le j \le k-1$.
Since there are $i+2$ paths attached to $V_{i}$, at least one of them, say $P$, 
has no vertex in the remaining vertices $b_{0}, \dots, b_{i}$.
The degree-1 vertex in $P$ has distance exactly $i$ from every vertex in $V_{i}$
and distance at least $i+1$ from every vertex not in $V(P) \cup V_{i}$.
Hence, $b_{i} \in V_{i}$.
Now we know that $b_{i} \in V_{i}$ for $2 \le i \le k-1$.
Let $u_{j} \in U$ be the vertex where $V_{0}$ and $V_{1}$ are attached to.
For $2 \le i \le k-1$, we have $\dist(b_{i}, u_{j}) \ge i$,
and thus $N_{i}[b_{i}]$ contains no vertex in $V_{0} \cup V_{1}$.
Since $N_{0}[b_{0}]$ will cover only one vertex,
$b_{1} = y$ and $b_{0} = w$ hold.

For $2 \le i \le k-1$, let $b_{i} = v_{h_{i}}^{(i)}$.
Since $N_{0}[b_{0}] = \{w\}$ and $N_{1}[b_{1}] = \{x,y,z\}$, we have
$U \subseteq \bigcup_{2 \le i \le k-1} N_{i}[v_{h_{i}}^{(i)}]$.
As we saw in the only-if case,
$u_{j} \in N_{i}[v_{h_{i}}^{(i)}]$ if and only if $u_{j} \in S_{h_{i}}$ for $1 \le j \le n$.
This implies that $N_{i}[v_{h_{i}}^{(i)}] \cap U = S_{h_{i}}$,
and thus $U = \bigcup_{2 \le i \le k-1} S_{h_{i}}$.
Therefore, the subfamily $\{S_{h_{2}}, S_{h_{3}}, \dots, S_{h_{k-1}}\} \subseteq \mathcal{S}$
of at most $k-2 = s$ subsets shows that $(U, \mathcal{S}, s)$ is a yes-instance of \textsc{Set Cover}.
\end{proof}

\begin{proof}
[Proof of Theorem~\ref{thm:w2-complete-k}]
By Lemma~\ref{lem:sc-to-gb},
the construction of $(G,k)$ from $(U,\mathcal{S},s)$ described above is a parameterized reduction 
from \textsc{Set Cover} parameterized by $s$ to \textsc{Graph Burning} parameterized by $k=s+2$.
Since \textsc{Set Cover} is W[2]-complete parameterized by $s$~\cite{DowneyF95i},
the W[2]-hardness follows.

The membership to W[2] can be shown by the following reduction to \textsc{Set Cover}.
Let $(G = (V,E), k)$ be an instance of \textsc{Graph Burning}.
We set $s = k$, $U = V \cup \{0, 1, \dots, k-1\}$, 
and $\mathcal{S} = \{N_{i}[v] \cup \{i\} \mid v \in V, 0 \le i \le k-1\}$.
This is just an undirected version of the proof by Janssen~\cite{Janssen20arxiv},
who showed the membership to W[2] for \textsc{Graph Burning} on directed graphs,
and the correctness can be shown in the same way.
\end{proof}

\begin{proof}
[Proof of Theorem~\ref{thm:no-poly-kernel-vc}]
The graph $G$ constructed above has an independent set $\bigcup_{2 \le i \le k-1} V_{i}$.
The vertices not belonging to this independent set form a vertex cover
of size $4 + (k-1)|U| + \sum_{2 \le i \le k-1} (i+2)(2i-2)$, which is a polynomial in $k = s+2$ and $|U|$.
By Lemma~\ref{lem:sc-to-gb},
the construction of $(G,k)$ from $(U,\mathcal{S},s)$ described above is a polynomial parameter transformation~\cite{BodlaenderTY11}
from \textsc{Set Cover} parameterized by $|U| + s$ to \textsc{Graph Burning} parameterized by vertex cover number.
Since \textsc{Set Cover} parameterized by $|U| + s$
does not admit polynomial kernels unless $\mathrm{NP} \subseteq \mathrm{coNP/poly}$~\cite{DomLS14},
the theorem holds.
\end{proof}

The reduction above also shows the W[2]-hardness parameterized by diameter
since the diameter of a connected graph is smaller than the square of its burning number~\cite{BonatoJR14,BonatoJR16}.
We further observe that the graph $G$ in the reduction is $P_{(4k-3)}$-free.
That is, $G$ does not contain a path of $4k-3$ vertices as an induced subgraph.
Let $P$ be an induced path in $G$. 
For $i \in \{2,\dots,k-1\}$, let $H_{i}'$ be the graph consists of $H_{i}$ and the paths from $U_{i}$ to $U$.
Since $U$ is a clique, there are at most two indices $i$ such that $P$ intersects $H_{i}' - U$.
We can see that $|V(P) \cap V(H_{i}')| \le 2i+1$ for each $i$, and thus $|V(P)| \le 2(k-1)+1 + 2(k-2)+1 = 4k-4$.
This implies the following $\mathrm{W[2]}$-hardness.
\begin{corollary}
\textsc{Graph Burning} on $P_{q}$-free graphs is $\mathrm{W[2]}$-hard parameterized by $q$.
\end{corollary}


\bibliography{burn-para}

\end{document}